\documentclass[letterpaper, 11pt]{article}
\usepackage{amsmath,amssymb,amsthm,xspace,fullpage}
\usepackage{algorithm}
\usepackage[noend]{algorithmic}
\newtheorem{theorem}{Theorem}[section]
\newtheorem{lemma}[theorem]{Lemma}
\newtheorem{definition}[theorem]{Definition}
\newtheorem{proposition}[theorem]{Proposition}
\newtheorem{corollary}[theorem]{Corollary}

\newcommand{\E}{\mathop{\mathbf{E}}}
\newcommand{\Var}{\mathop{\mathbf{Var}}}
\newcommand{\bbA}{\mathbb{A}}
\newcommand{\bbB}{\mathbb{B}}
\newcommand{\bbF}{\mathbb{F}}
\newcommand{\bbH}{\mathbb{H}}
\newcommand{\bfa}{\mathbf{a}}

\newcommand{\biw}{\boldsymbol{w}}
\newcommand{\caA}{\mathcal{A}}
\newcommand{\caB}{\mathcal{B}}
\newcommand{\caG}{\mathcal{G}}
\newcommand{\caH}{\mathcal{H}}
\newcommand{\caI}{\mathcal{I}}
\newcommand{\caJ}{\mathcal{J}}

\newcommand{\caS}{\mathcal{S}}

\newcommand{\caV}{\mathcal{V}}
\newcommand{\dist}{\mathrm{dist}}
\newcommand{\maj}{\textrm{maj}}

\newcommand{\HornSAT}{\textsf{Horn-SAT}\xspace}
\newcommand{\kSAT}{\textsf{$k$-SAT}\xspace}
\newcommand{\twSAT}{\textsf{2SAT}\xspace}
\newcommand{\thSAT}{\textsf{3SAT}\xspace}
\newcommand{\CSP}[1]{\textsf{CSP}(#1)\xspace}
\newcommand{\HOM}[1]{\textsf{HOM}(#1)\xspace}
\newcommand{\LHOM}[1]{\textsf{LHOM}(#1)\xspace}
\newcommand{\Pol}[1]{\textsf{Pol}(#1)\xspace}
\newcommand{\Inv}{\mathsf{Inv}}
\newcommand{\Alg}[1]{\textsf{Alg}(#1)\xspace}
\newcommand{\Str}[1]{\textsf{Str}(#1)\xspace}
\newcommand{\Term}[1]{\textsf{Term}(#1)\xspace}
\newcommand{\ttmin}{\textsf{$(2,3)$-Minimality}\xspace}
\newcommand{\klmin}{\textsf{$(k,\ell)$-Minimality}\xspace}

\title{Testing List $H$-Homomorphisms}
\author{Yuichi Yoshida\thanks{Supported by MSRA Fellowship 2010.}\\\\
  School of Informatics, Kyoto University, and\\ Preferred Infrastructure, Inc.\\yyoshida@kuis.kyoto-u.ac.jp}
\begin{document}
\setcounter{page}{0}
\maketitle

\begin{abstract}
  Let $H$ be an undirected graph.
  In the \textsf{List $H$-Homomorphism Problem},
  given an undirected graph $G$ with a list constraint $L(v) \subseteq V(H)$ for each variable $v \in V(G)$,
  the objective is to find a list $H$-homomorphism $f:V(G) \to V(H)$,
  that is,
  $f(v) \in L(v)$ for every $v \in V(G)$ and $(f(u),f(v)) \in E(H)$ whenever $(u,v) \in E(G)$.

  We consider the following problem:
  given a map $f:V(G) \to V(H)$ as an oracle access,
  the objective is to decide with high probability whether $f$ is a list $H$-homomorphism or \textit{far} from any list $H$-homomorphisms.
  The efficiency of an algorithm is measured by the number of accesses to $f$.
  
  In this paper, we classify graphs $H$ with respect to the query complexity for testing list $H$-homomorphisms and show the following trichotomy holds:
  (i)  List $H$-homomorphisms are testable with a constant number of queries if and only if $H$ is a reflexive complete graph or an irreflexive complete bipartite graph.
  (ii) List $H$-homomorphisms are testable with a sublinear number of queries if and only if $H$ is a bi-arc graph.
  (iii) Testing list $H$-homomorphisms requires a linear number of queries if $H$ is not a bi-arc graph.
\end{abstract}

\newpage

\section{Introduction}
For two graphs $G = (V(G),E(G))$ and $H = (V(H),E(H))$,
a map $f: V(G) \to V(H)$ is called a \textit{homomorphism} from $G$ to $H$ if $(f(u),f(v)) \in E(H)$ whenever $(u,v) \in E(G)$.
In the \textsf{$H$-Homomorphism Problem} ($\HOM{H}$ for short),
given an undirected graph $G$,
the objective is to decide whether there exists a homomorphism from $G$ to $H$.
It is well known that $\HOM{H}$ is in \textbf{P} if $H$ is a bipartite graph and in \textbf{NP-Complete} if $H$ is not a bipartite graph~\cite{barto2010cyclic,bulatov2005h,hell1990complexity}.

\textsf{List $H$-Homomorphism Problem} ($\LHOM{H}$ for short) is a variant of \textsf{$H$-Homomorphism Problem},
in which we are also given a list $L(v) \subseteq V(H)$ for each vertex $v$ in $G$.
A map $f:V(G) \to V(H)$ is called a \textit{list-homomorphism} from $G$ to $H$ if $f$ is a homomorphism from $G$ to $H$ and $f(v) \in L(v)$ for every $v \in V(G)$.
The objective is to decide whether there exists a list-homomorphism $f$ from $G$ to $H$.
There are many results on the relationship between the graph $H$ and the computational complexity of $\LHOM{H}$~\cite{egri2010complexity,feder1998list,feder1999list,feder2003bi}.
In particular, $\LHOM{H}$ is in \textbf{P} iff $H$ is a bi-arc graph~\cite{feder2003bi}.

In this paper, 
we consider testing list-homomorphisms.
See~\cite{goldreich10introduction,ron2009algorithmic} for surveys on \textit{property testing}.
In our setting, a map $f$ is given as an oracle access, i.e., 
the oracle returns $f(v)$ if we specify a vertex $v \in V(G)$.
A map $f$ is called \textit{$\epsilon$-far} from list-homomorphisms if we must modify at least an $\epsilon$-fraction of $f$ to make $f$ a list-homomorphism.
An algorithm is called a \textit{tester} for $\LHOM{H}$ if it accepts with probability at least $2/3$ if $f$ is a list-homomorphism from $G$ to $H$ and rejects with probability at least $2/3$ if $f$ is $\epsilon$-far from list-homomorphisms.
The efficiency of an algorithm is measured by the number of accesses to the oracle $f$.
When we say that a query complexity is constant/sublinear/linear,
it always means constant/sublinear/linear in $|V(G)|$,
i.e., the domain size of $f$.
We can assume that there exists a list-homomorphism from $G$ to $H$.
If otherwise, we can reject immediately without any query.

In this paper, we completely classify graphs $H$ with respect to the query complexity for testing $\LHOM{H}$.
Our result consists of the following two theorems.
\begin{theorem}\label{thr:constant-list}
  $\LHOM{H}$ is testable with a constant number of queries iff $H$ is an irreflexive complete bipartite graph or a reflexive complete graph.
\end{theorem}
\begin{theorem}\label{thr:sublinear-list}
  $\LHOM{H}$ is testable with a sublinear number of queries iff $H$ is a bi-arc graph.
\end{theorem}
The central question in the area of property testing is to classify properties into the following three categories:
properties testable with a constant/sublinear/linear number of queries.
Our result first establishes such a classification for a natural and general combinatorial problem.

We note that,
from Theorem~\ref{thr:sublinear-list} and results given by~\cite{feder2003bi},
 $\LHOM{H}$ is testable with a sublinear number of queries iff $\LHOM{H}$ is in \textbf{P}.
However, it is not clear whether there is a computational class corresponding to properties testable with a constant number of queries.

To obtain our results,
we exploit universal algebra, which is now a common tool to study computational complexity of constraint satisfaction problems (see, e.g.,~\cite{jeavons1997closure}).
Another contribution of this paper is showing that universal algebraic approach is quite useful in the setting of property testing.

\paragraph{Related works:}
It is rare that we succeed to obtain characterizations of properties testable with a constant/sublinear number of queries.
The only such a characterization we are aware of is one for graph properties in the \textit{dense model}~\cite{goldreich1998property}.
In this model,
it is revealed that Szemer{\'e}di's regularity lemma~\cite{szemerdi1975regular} plays a crucial role~\cite{alon2006combinatorial}.
Roughly speaking, the regularity lemma gives the constant-size sketch of a graph.
It turns out that a property is testable in the dense model with a constant number of queries iff the property is well-approximated by the union of constant number of sketches~\cite{alon2006combinatorial}.
However, no characterization is known for properties testable with a sublinear number of queries.
Similarly, for properties on Boolean functions,
several partial classifications on constant-time testability are known~\cite{bhattacharyya2010unified,kaufman2008algebraic}.

Let $\caB$ be a relational structure (see Section~\ref{sec:pre} for the definition).
In \textsf{Constraint Satisfaction Problem} over $\caB$ ($\CSP{\caB}$ for short),
given another relational structure $\caA$,
the objective is to find a homomorphism from $\caA$ to $\caB$.
There have been a lot of research on classifying $\caB$ with respect to the computational complexity of \CSP{$\caB$} (see, e.g., \cite{bulatov2008recent,hell2008colouring}).
We can see that $\HOM{H}$ and $\LHOM{H}$ are special cases of \textsf{CSP},
and there are also many results classifying $H$ with respect to the computational complexity of $\HOM{H}$ and $\LHOM{H}$ (see, e.g., \cite{hell2003algorithmic,hell2008colouring}).

Testing homomorphisms on \kSAT is already studied~\cite{ben2006some,fischer2002monotonicity}.
In \kSAT,
given a CNF formula for which each clause consists of $k$ literals,
the objective is to find an assignment to variables so as to satisfy all the clauses.
The problem \kSAT coincides with \CSP{$\caB$} for some appropriate relational structure $\caB$.
Testing homomorphisms for \kSAT can be restated as follows:
Given an assignment,
test whether the assignment is a satisfying assignment or far from any satisfying assignments.
It is known that \twSAT is testable with $O(\sqrt{n})$ queries~\cite{fischer2002monotonicity}
and testing \thSAT requires $\Omega(n)$ queries~\cite{ben2006some},
where $n$ is the number of variables in an input CNF formula.
We can see that our result extends those results to a large family of \textsf{CSP}s.

Given a relational structure $\caB$,
it is also natural to ask whether a relational structure $\caA$ has a homomorphism to $\caB$ or far from having homomorphisms.
Here, $\caA$ is given as an oracle access.
There are two major models in this setting, i.e., the \textit{dense model} and the \textit{bounded-degree model}.
In the dense model,
it is known that any CSP is testable in constant time~\cite{alon2003testing}.
In the bounded-degree model,
it is known that \HornSAT is testable in constant time~\cite{yoshida2011testing,yoshida2011optimal} and the number of queries needed to test \twSAT is $\widetilde{\Theta}(\sqrt{n})$~\cite{goldreich1999sublinear} where $n$ is the number of variables in an input structure $\caA$.

\paragraph{Organizations:}
In Section~\ref{sec:pre},
we introduce definitions used throughout this paper.
Sections~\ref{sec:constant-upper} and~\ref{sec:constant-lower} shows ``if'' and ``only if'' part of Theorem~\ref{thr:constant-list}, respectively.
Similarly, Sections~\ref{sec:sublinear-upper} and~\ref{sec:sublinear-lower} shows ``if'' and ``only if'' part of Theorem~\ref{thr:sublinear-list}, respectively.

\section{Preliminaries}\label{sec:pre}
Let $\biw:A \to \mathbb{R}$ be a weight function such that $\sum_{a \in V}\biw(a) =1$.
Then, by $a \sim \biw$, we mean that we pick an element $a \in A$ with probability $\biw(a)$.
For a function $f:A\to B$ and $A' \subseteq A$,
we define $f|_{A'}:A' \to A$ as the function whose domain is restricted to $A'$.
We also define $\biw|_{A'}: A' \to \mathbb{R}$ as $\biw|_{A'}(a) = \biw(a) / \sum_{a' \in A'}\biw(a')$.
Note that $\biw|_{A'}$ satisfies $\sum_{a' \in A'}\biw|_{A'}(a') = 1$.

\paragraph{Relational structures, polymorphisms and algebras:}
A \textit{vocabulary} $\tau$ is a finite set of \textit{relational symbols}; 
each symbol has an associated \textit{arity}.
A (finite) relational structure $\caA$ with vocabulary $\tau$ consists of a finite set $A$, 
its \textit{universe}, and for every relational symbol $R \in \tau$ of arity $n$,
an $n$-ary relation $R^{\caA}$ on $A$,
the \textit{interpretation} of $R$ by $\caA$.
A \textit{homomorphism} of a structure $\caA$ to a structure $\caB$ with the same vocabulary $\tau$ is a mapping $\varphi : A \to B$ from the universe of $\caA$ to the universe of $\caB$ such that for each ($n$-ary) relational symbol $R \in \tau$ and any tuple $(a_1,\ldots,a_n) \in R^{\caA}$ the tuple $(\varphi(a_1),\ldots,\varphi(a_n)) \in R^{\caB}$.
For relational structure $\caB$,
we define $\HOM{\caB}$ as the problem in which,
given another relational structure $\caA$,
the objective is to find a homomorphism from $\caA$ to $\caB$.
We denote by $|\caA|$ the size of universe of $\caA$ and denote by $\|\caA\|$ the number of tuples in its relations.
For brevity, we use a capital letter to denote the universe of the corresponding relational structure (e.g., $A$ denotes the universe of $\caA$).

Let $R$ be a relation on a set $A$,
An ($n$-ary) operation $f$ on the same set is said to be a \textit{polymorphism} of $R$ if for any tuples $\bfa_1,\ldots,\bfa_n \in R$,
the tuple $f(\bfa_1,\ldots,\bfa_n)$ obtained by applying $f$ component-wise also belongs to $R$.
The relation $R$ is called an \textit{invariant} of $f$.
An operation $f$ is a polymorphism of a relational structure $\caA$ if it is a polymorphism of each relation of the structure.
The set of all polymorphism of $\caA$ is denoted by $\Pol{\caA}$.
From a collection $C$ of operations $\Inv(C)$ denotes the set of invariants of all operations from $C$.

An algebra is a pair $\bbA = (A; F)$ consisting of a set $A$, the \textit{universe} of $\bbA$,
and a set $F$ of operations, the \textit{basic operations} of $\bbA$.
Operations that can be obtained from the basic operations of $\bbA$ and the projection operations on $\bbA$,
that is operations of the form $f(x_1,\ldots,x_n ) = x_i$, 
by means of compositions are called \textit{term operations} of $\bbA$.
$\Term{A}$ denotes the set of all term operations of $\bbA$.
Any relational structure $\caA$ can be associated with an algebra $\Alg{\caA} = (A; \Pol{\caA})$.
Conversely, any algebra $\bbA = (A; F)$ corresponds to a class of relational structures $\Str{\bbA}$ that includes all the structures $\caA$ with universe $A$ and $\Term{\bbA} \subseteq \Pol{\caA}$.

An operation $f$ is called \textit{idempotent} if $f(x_1,\ldots,x_k) \in \{x_1,\ldots,x_k\}$ for any $x_1,\ldots,x_k$ where $k$ is the arity of $f$.
An algebra $\bbA$ is called \textit{idempotent} if any basic operation (thus, term operation) is idempotent.

\paragraph{Graph homomorphisms:}
We often identify a graph $H = (V(H),E(H))$ with a relational structure $\caH$ consisting of a universe $V(H)$ and a binary relation $E(H)$.
In particular, $|\caH| = |V(H)|$ and $\|\caH\| = |E(H)|$.
Note that problem $\HOM{H}$ defined with graph terminologies and $\HOM{\caH}$ defined with relational structures coincide.
Similarly, $\LHOM{H}$ can be restated with relational structures.
Let $\caH$ be the relational structure associated with $H$ and define $\caH^L$ as the relational structure obtained from $\caH$ by adding all unary relations $S^\caH \subseteq V(H)$.
Then, $\LHOM{\caH}$ coincides with $\HOM{\caH^L}$.
Let $G$ be a graph with list constraints $\{L(v)\}_{v \in V(G)}$.
Let $\caG$ be the relational structure corresponding to $G$.
Then, we can define another relational structure $\caG^L$ obtained from $\caG$ by adding unary relations $S^\caG \subseteq V(G)$.
Here, for each (unary) tuple $(v) \in S^\caG$,
there is a corresponding constraint $L(v) = S^\caH$.
Finally, we define $\bbH^L = \Alg{\caH^L}$ for a graph $H$.
Note that the algebra $\bbH^L$ is idempotent since $\caH^L$ contains unary relations $S_v = \{(v)\}$ for every $v \in V(H)$, and an operation that has $\{S_v\}_{v \in V(H)}$ as invariants must be idempotent.

\paragraph{Testing homomorphisms:}
Let $\caA$ and $\caB$ be relational structures and $\biw:A \to \mathbb{R}$ be a weight function with $\sum_{a \in A}\biw(a) = 1$.
The distance between two functions $f,f':A\to B$ is defined as $\dist(f,f') = \Pr\limits_{a \sim \biw}[f(a) \neq f'(a)]$.
For a function $f:A \to B$,
we define $\dist_{\caB}(f) = \min_{f'}\dist(f,f')$ where $f'$ is a homomorphism from $\caA$ to $\caB$.
We call a map $f$ \textit{$\epsilon$-far} if $\dist_{\caB}(f) \geq \epsilon$.
For subset $A' \subseteq A$,
we define $\dist_{\caB}(f|_{A'})$ similarly using the weight function $\biw|_{A'}$.

We consider testing homomorphisms to a relational structure $\caB$.
An input is $(\caA,\biw,f)$ where $\caA$ is a relational structure, $\biw:A \to \mathbb{R}$ is a weight function with $\sum_{a \in A}\biw(a) = 1$, and $f:A \to B$ is a map.
The map $f$ is given as an oracle access.
Thus, by specifying $a \in A$, the oracle returns the value of $f(a)$.
\begin{definition}
  An algorithm is called a tester for $\HOM{\caB}$ if,
  given an input $(\caA,\biw,f)$,
  it accepts with probability at least $2/3$ when $f$ is a homomorphism from $\caA$ to $\caB$,
  and rejects with probability at least $2/3$ when $f$ is $\epsilon$-far from homomorphisms.
\end{definition}
A tester is called a \textit{one-sided error tester} if it always accepts an input $(\caA,\biw,f)$ if $f$ is a homomorphism.
An algorithm is called having query complexity $q(n,m,\epsilon)$ if,
given an input $(\caA,\biw,f)$,
it queries at most $q(|\caA|,\|\caA\|,\epsilon)$ times.
We always assume that $q(n,m,\epsilon)$ is an increasing function on $n,m$ and a decreasing function on $\epsilon$.

We can similarly define testers and testability for $\HOM{H}$ and $\LHOM{H}$ for a graph $H$ since they have equivalent formalizations using relational structures.
For convenience,
we write an input for $\HOM{H}$ and $\LHOM{H}$ as $(G,\biw,f)$ and $(G,L,\biw,f)$, respectively,
where $G$ is a graph, $\{L(v)\}_{v \in V(G)}$ is a set of list constraints, $\biw$ is a weight function, and $f$ is a map given as an oracle access.

We finally define bi-arc graphs for reference though we do not use the definition in our proof.
Let $C$ be a circle with two specified points $p$ and $q$.
A bi-arc is a pair of arcs $(N,S)$ such that $N$ contains $p$ but not $q$ and $S$ contains $q$ but not $p$.
A graph $H=(V,E)$ is a \textit{bi-arc graph} if there is a family of bi-arcs $\{(N_x,S_x) \mid x \in V\}$ such that,
for every $x,y\in V$,
the following holds:
(i) if $x$ and $y$ are adjacent, then neither $N_x$ intersects $S_y$ nor $N_y$ intersects $S_x$,
and $(ii)$ if $x$ is not adjacent to $y$ then both $N_x$ intersects $S_y$ and $N_y$ intersects $S_x$.
An undirected graph is called \textit{reflexive} if every vertex has a loop and \textit{irreflexive} if no vertex has loop.
It is known that a reflexive graph is bi-arc iff it is an interval graph,
and that an irreflexive graph is bi-arc iff it is bipartite and its complement is an circular-arc graph~\cite{feder2003bi}.

\section{Graphs Testable with a Constant Number of Queries}\label{sec:constant-upper}

In this section, we show the following lemma, which is the ``if'' part of Theorem~\ref{thr:constant-list}.
\begin{lemma}\label{lmm:constant-upper}
  When $H$ is an irreflexive complete bipartite graph or a reflexive complete graph,
  then there exists a one-sided error tester for $\LHOM{H}$ with query complexity $O(1/\epsilon^2)$.
\end{lemma}

Let $(G,L,\biw,f)$ be an input for $\LHOM{H}$.
For a vertex $v \in V(G)$, let $C(v)$ be the connected component containing $v$.

\begin{proposition}\label{prp:cc}
  Let $(G,L,\biw,f)$ be an input for $\LHOM{H}$.
  Suppose that $\dist_H(f) \geq \epsilon$.
  Then,  $\E\limits_{v \sim \biw}[\dist_H(f|_{C(v)})] \geq \epsilon$.
  \qed
\end{proposition}
\begin{proof}
  Let $\overline{f}$ be a list-homomorphism closest to $f$ and $\epsilon_v = \dist(f|_{C(v)},\overline{f}|_{C(v)})$.
  It is clear that $\epsilon_v = \dist_H(f|_{C(v)})$.
  Since $\dist_H(f) \geq \epsilon$,
  we have $\E\limits_{v \sim \biw}[\dist_H(f|_{C(v)})] = \E\limits_{v \sim \biw}[\epsilon_v] \geq \epsilon$.
\end{proof}

\begin{corollary}\label{crl:cc-testable}
  Suppose that there exists a one-sided error tester $\caA$ for $\LHOM{H}$ with query complexity $q(\epsilon)$  if an input graph is restricted to be connected.
  Then, there exists a one-sided error tester $\caA'$ for $\LHOM{H}$ with query complexity $q(\epsilon)/\epsilon$ for any input graph.
\end{corollary}
\begin{proof}
  Let $(G,L,\biw,f)$ be an input.
  Our algorithm $\caA'$ is as follows:
  Let $S$ be a set of $\Theta(1/\epsilon)$ vertices chosen from $G$ according to $\biw$.
  For each vertex $v \in S$,
  we run $\caA$ with an error parameter $\epsilon$ on the input whose domain is restricted to $C(v)$.
  We reject if $\caA$ rejects for some $v \in S$.
  We accept otherwise.

  It is easy to see that the above algorithm always accepts when $f$ is a list-homomorphism,
  and the query complexity is at most $q(\epsilon)/\epsilon$.

  Suppose that $\dist_H(f) \geq \epsilon$.
  From Proposition~\ref{prp:cc},
  we have $\E\limits_{v \sim \biw}[\dist_H(f|_{C(v)})] \geq \epsilon$.
  Thus, $\Pr\limits_{v \sim \biw}[\dist_H(f|_{C(v)}) \geq \epsilon] \geq \epsilon$ holds.
  For a vertex $v$ with $\dist_H(f|_{C(v)}) \geq \epsilon$,
  the algorithm $\caA$ rejects with probability at least $2/3$.
  Thus, $\caA'$ rejects with probability at least $1 - (1-2/3 \cdot \epsilon)^{\Theta(1/\epsilon)} \geq 2/3$.
\end{proof}

\begin{lemma}\label{lmm:list-testable}
  Suppose that there exists a one-sided error tester $\caA$ for $\LHOM{H}$ with query complexity $q(\epsilon)$ if an input map is restricted to satisfy list constraints.
  Then, there also exists a one-sided error tester $\caA'$ for $\LHOM{H}$ with query complexity $O(1/\epsilon) + q(\epsilon/2)$ for any input.
\end{lemma}
\begin{proof}
  Let $(G,L,\biw,f)$ be an input.
  We define $f' : G \to H$ as 
  \begin{eqnarray*}
    f'(v) = \begin{cases}
      f(v) & \text{if } f(v) \in L(v) \\
      \text{any element in } L(v) & \text{otherwise} 
    \end{cases}
  \end{eqnarray*}
  Note that we can compute $f'(v)$ by querying $f$ once, i.e., $f(v)$.

  Our algorithm $\caA'$ is as follows:
  Let $S$ be a set of $\Theta(1/\epsilon)$ vertices chosen from $G$ according to $\biw$.
  We reject if $f(v) \not \in L(v)$ for some $v \in S$.
  If otherwise,
  we simply return the output by $\caA$ running on $f'$ with an error parameter $\epsilon / 2$.

  It is easy to see that the above test always accepts when $f$ is a list-homomorphism,
  and the query complexity is at most $O(1/\epsilon) + q(\epsilon/2)$.
  
  Suppose that $\dist_H(f) \geq \epsilon$.
  If $\dist(f,f') \geq \epsilon/2$,
  then $\caA'$ rejects with probability at least $1 - (1-\epsilon/2)^{\Theta(1/\epsilon)} \geq 2/3$
  when checking $f(v) \in L(v)$ for $v \in S$.
  If $\dist(f,f') < \epsilon/2$,
  from the triangle inequality,
  we have $\dist_H(f') \geq \epsilon/2$.
  Thus, $\caA'$ rejects with probability at least $2/3$.
\end{proof}

\begin{lemma}\label{lmm:rcg}
  Let $H$ be a reflexive complete graph.
  Then, there exists a one-sided error tester for $\LHOM{H}$ with query complexity $O(1/\epsilon)$.
\end{lemma}
\begin{proof}
  Let $(G,L,\biw,f)$ be an input and assume that $f$ satisfies list constraints.
  Then, we can always accept since any map is a list-homomorphism to $H$.
  The lemma follows from Lemma~\ref{lmm:list-testable}.
\end{proof}

\begin{lemma}\label{lmm:ktw-testable}
  Let $K_2$ be an irreflexive complete graph with two vertices.
  There exists a one-sided error tester $\caA$ for $\LHOM{K_2}$ with query complexity $O(1/\epsilon^2)$.
\end{lemma}
\begin{proof}
  Let $(G,L,\biw,f)$ be an input and assume that $G$ is connected and $f$ satisfies list constraints.
  Note that $G$ must be bipartite in order to have a homomorphism to $H$, 
  and let $V_1 \cup V_2$ be the bipartition of $G$.
  Let $a,b$ be two vertices in $K_2$.
  Then, we have two homomorphisms, i.e., $f^1$ and $f^2$ where $f^1|_{V_1} \equiv a, f^1|_{V_2} \equiv b$,
  and $f^2|_{V_1} \equiv b, f^2|_{V_2} \equiv a$.
  
  Our algorithm $\caA$ is as follows:
  Let $S_1$ (resp., $S_2$) be a set of $\Theta(1/\epsilon)$ vertices chosen from $V_1$ (resp., $V_2$) according to $\biw|_{V_1}$ (resp., $\biw|_{V_2}$). 
  Then, we check $f(u) = f(v)$ for every $u,v \in S_1$,
  $f(u) = f(v)$ for every $u,v \in S_2$,
  and $f(u) \neq f(v)$ for every $u \in S_1,v \in S_2$.
  We reject if any of them do not hold.
  We accept otherwise.

  It is easy to see that the above test always accepts when $f$ is a list-homomorphism,
  and the query complexity is $O(1/\epsilon)$.
   
  Suppose that $\dist_H(f) \geq \epsilon$.
  Note that $f^1$ or $f^2$ must satisfy list constraints.
  We assume below that both $f^1$ and $f^2$ satisfy list constraints.
  The analysis is similar when either of them does not satisfy list constraints.
  
  Since $f$ is $\epsilon$-far from $f^1$,
  we have $\sum_{v \in V_1,f(v) = b}\biw(v)  \geq \epsilon/2 $ or $\sum_{v \in V_2,f(v) = a}\biw(v)  \geq \epsilon/2$.
  Similarly,
  we have $\sum_{v \in V_1,f(v) = a}\biw(v)  \geq \epsilon/2 $ or $\sum_{v \in V_2,f(v) = b}\biw(v)  \geq \epsilon/2$.
  In any case,
  the probability that we reject is at least $1 - 2(1-\epsilon/2)^{\Theta(1/\epsilon)} \geq 2/3$.
  
  From Corollary~\ref{crl:cc-testable} and Lemma~\ref{lmm:list-testable},
  $\LHOM{K_2}$ is testable with $O(1/\epsilon^2)$ queries.
\end{proof}

Now, we show that any complete bipartite graph is testable with a constant number of queries.
For two graphs $G$ and $H$,
we call a map $f : V(G) \to V(H)$ a \textit{full-homomorphism} if 
$(u,v) \in E(G)$ iff $(f(u),f(v)) \in E(H)$.
The difference from a homomorphism is that we must have $(f(u),f(v)) \not \in E(H)$ when $(u,v) \not \in E(G)$.
\begin{proposition}\label{prp:full-hom}
  Let $h$ be a full-homomorphism from $H$ to $H'$.
  For a graph $G$ and a homomorphism $f'$ from $G$ to $H'$,
  let $f:V(G) \to V(H)$ be a map such that $f(v)$ is any element in $h^{-1} (f'(v))$.
  Then, $f$ is also a homomorphism from $G$ to $H$.
\end{proposition}
\begin{proof}
  Suppose that $f$ is not a homomorphism.
  Then, there exist $u,v \in V(G)$ such that $(u,v) \in E(G)$ while $(f(u),f(v)) \not \in E(H)$.
  Then, we have $(f'(u),f'(v)) = (h(f(u)),h(f(v))) \not \in E(H')$ since $h$ is a full-homomorphism,
  which is contradicting the fact that $f'$ is a homomorphism.
\end{proof}

\begin{lemma}\label{lmm:full-hom-testable}
  Let $H$ be a graph and suppose that there exists a full-homomorphism $h$ from $H$ to $H'$.
  If there exists a one-sided error tester for $\LHOM{H'}$ with query complexity $q(\epsilon)$, 
  then there exists a one-sided error tester for $\LHOM{H}$ with query complexity $O(1/\epsilon) + q(\epsilon/2)$.
\end{lemma}
\begin{proof}
  Let $(G,L,\biw,f)$ be an input and assume that $f$ satisfies list constraints.
  We define $L' = h \circ L$ and $f' = h\circ f$.
  Then, we run the tester for $\LHOM{H'}$ on an input $(G,L',\biw,f')$ with an error parameter $\epsilon$.

  If $f$ is a list-homomorphism,
  then $f'$ is also a list-homomorphism,
  and the tester always accepts.

  Suppose that $\dist_H(f) \geq \epsilon$ and let $\overline{f'}$ be the list-homomorphism closest to $f'$.
  We define $\widetilde{f}:V(G) \to V(H)$ as follows.
  \begin{eqnarray*}
    \widetilde{f}(v) = \begin{cases}
      f(v)  & \text{if } \overline{f'}(v) = f'(v), \\
      \text{any element in } L(v) \cap h^{-1}(\overline{f'}(v)) & \text{otherwise}.\\
    \end{cases}
  \end{eqnarray*}
  Note that, in the latter case,
  $L(v) \cap h^{-1}(\overline{f'}(v))$ is not empty since $\overline{f'}(v) \in L'(v) = h(L(v))$.
  Thus, $\widetilde{f}$ is well-defined.
  We can easily see that $\widetilde{f}$ satisfies list constraints from the construction. 
  When $\overline{f'}(v) = f'(v)$,
  we have $f(v) \in h^{-1}(\overline{f'}(v))$.
  Thus, 
  $\widetilde{f}$ is a list-homomorphism from Proposition~\ref{prp:full-hom}.
  It means that $\dist(f',\overline{f'}) \geq \dist(f,\widetilde{f}) \geq \dist_H(f) \geq \epsilon$.
  Thus, the tester rejects with probability at least $2/3$.

  We have shown that $\LHOM{H}$ is testable with $q(\epsilon)$ queries when an input is restricted to satisfy list constraints.
  The lemma follows from Lemma~\ref{lmm:list-testable}.
\end{proof}

\begin{lemma}\label{lmm:icbg}
  Let $H$ be an irreflexive complete bipartite graph.
  Then, there exists a one-sided error tester for $\LHOM{H}$ with query complexity $O(1/\epsilon^2)$.
\end{lemma}
\begin{proof}
  Since there exists a full-homomorphism from $H$ to $K_2$,
  $\LHOM{H}$ is testable with $O(1/\epsilon^2)$ queries from Lemmas~\ref{lmm:ktw-testable} and~\ref{lmm:full-hom-testable}.
\end{proof}
\begin{proof}[Proof of Lemma~\ref{lmm:constant-upper}]
  The claim immediately follows from Lemmas~\ref{lmm:rcg} and~\ref{lmm:icbg}.
\end{proof}

\section{Graphs Not Testable with a Constant Number of Queries}\label{sec:constant-lower}
In this section, we show the following lemma, 
which is the ``only if'' part of Theorem~\ref{thr:constant-list}.
\begin{lemma}\label{lmm:constant-lower}
  If $H$ is neither an irreflexive complete bipartite graph nor a reflexive complete graph, 
  testing $\LHOM{H}$ requires $\Omega(\frac{\log n}{\log \log n})$ queries.
\end{lemma}

The following is immediate since we can freely restrict the range of a map by list constraints.
\begin{proposition}\label{prp:induced}
  Let $H$ be a graph and $H'$ be an induced subgraph of $H$.
  If testing $\LHOM{H'}$ requires $q$ queries,
  then testing $\LHOM{H}$ also requires $q$ queries.
\end{proposition}

The following lemma is implicitly stated in~\cite{fischer2002monotonicity} when showing lower bounds for testing \twSAT.
\begin{lemma}[\cite{fischer2002monotonicity}]\label{lmm:reachability-is-hard}
  Let $H$ be a graph and $G$ be a graph with list constraints $\{L(v)\}_{v \in V(G)}$.
  For vertices $u,v \in V(G)$ with $u\neq v$, 
  let $R = \{(f(u),f(v)) \mid f \text{ is a list-homomorphism from } G \text{ to }H\}$.
  If $R = \{(a,c),(b,c),(b,d)\}$ for vertices $a,b,c,d \in H, a \neq b, c \neq d$,
  then testing $\LHOM{H}$ requires $\Omega(\frac{\log n}{\log \log n})$ queries.
\end{lemma}
We also use the following lemma, which is a special case of Lemma~\ref{lmm:sublinear-lower} (see Section~\ref{sec:sublinear-lower}).
\begin{lemma}\label{lmm:triangle-is-hard}
  Let $K_3$ be an irreflexive complete graph with three vertices.
  Then, testing $\LHOM{K_3}$ requires $\Omega(n)$ queries.
\end{lemma}
Note that $K_3$ is not a bi-arc graph.

The following two lemmas deal with reflexive graphs and irreflexive graphs, respectively.
\begin{lemma}\label{lmm:reflexive-lower}
  Let $H$ be a reflexive graph.
  If $H$ is not a complete graph, 
  testing $\LHOM{H}$ requires $\Omega(\frac{\log n}{\log \log n})$ queries.
\end{lemma}
\begin{proof}
  Let $P = (\{a,b,c\};\{(a,a),(b,b),(c,c),(a,b),(b,c)\})$ be a reflexive path of length $2$ and $G = (\{u,v\}; \{(u,v)\})$ be an irreflexive edge with list constraints $L(u) = \{a,b\} $ and $L(v) = \{b,c\}$.
  It is easy to see that the relation $R$ in Lemma~\ref{lmm:reachability-is-hard} becomes $R = \{(a,b),(b,b),(b,c)\}$.
  It follows that testing $\LHOM{P}$ requires $\Omega(\frac{\log n}{\log \log n})$ queries.
  From Proposition~\ref{prp:induced},
  if $\LHOM{H}$ is testable with a constant number of queries,
  $H$ must have a diameter $1$,
  implying that $H$ is a reflexive complete graph.
\end{proof}

\begin{lemma}\label{lmm:irreflexive-lower}
  Let $H$ be an irreflexive graph.
  If $H$ is not a complete bipartite graph,
  testing $\LHOM{H}$ requires $\Omega(\frac{\log n}{\log \log n})$ queries.
\end{lemma}
\begin{proof}
  Let $P=(\{a,b,c,d\}; \{(a,b),(b,c),(c,d)\})$ be an irreflexive path of length $3$ and $G = (\{u,v\}; \{(u,v)\})$ be an irreflexive edge with list constraints $L(u) = \{a,c\}$ and $L(v) = \{b,d\}$.
  It is easy to see that the relation $R$ in Lemma~\ref{lmm:reachability-is-hard} becomes $R = \{(a,b),(c,b),(c,d)\}$.
  It follows that testing $\LHOM{P}$ requires $\Omega(\frac{\log n}{\log \log n})$ queries.
  Also, for an irreflexive triangle $T$,
  testing $\LHOM{T}$ requires $\Omega(n)$ queries from Lemma~\ref{lmm:triangle-is-hard}.

  Thus, if $\LHOM{H}$ is testable with a constant number of queries,
  $H$ must not have a path of length $3$ or a triangle as induced subgraphs from Proposition~\ref{prp:induced}.
  Thus, $H$ must be a complete bipartite graph.
\end{proof}

\begin{proof}[Proof of Lemma~\ref{lmm:constant-lower}]
  Let $P = (\{a,b\}; \{(a,b),(b,b)\})$ be an edge with a loop and $G = (\{u,v\}; \{(u,v)\})$ be an irreflexive edge with list constraints $L(u) = L(v) = \{a,b\}$.
  Then, the relation $R$ in Lemma~\ref{lmm:reachability-is-hard} becomes $\{(a,b),(b,b),(b,a)\}$,
  and it follows that testing $\LHOM{P}$ requires $\Omega(\frac{\log n}{\log \log n})$ queries.
  Thus, if $\LHOM{H}$ is testable with a constant number of queries,
  $H$ must be reflexive or irreflexive from Proposition~\ref{prp:induced}.
  Then, the lemma follows from Lemmas~\ref{lmm:reflexive-lower} and~\ref{lmm:irreflexive-lower}.
\end{proof}

\section{Graphs Testable with a Sublinear Number of Queries}\label{sec:sublinear-upper}
In this section, we show the following lemma,
which is the ``if'' part of Theorem~\ref{thr:sublinear-list}.
\begin{lemma}\label{lmm:sublinear-upper}
  Let $H$ be a bi-arc graph.
  Then, there exists a one-sided error tester for $\LHOM{H}$ with query complexity $O(\sqrt{n/\epsilon})$. 
\end{lemma}

We first describe a propagation algorithm to solve \textsf{CSP}s.
Let $\caA$ and $\caB$ be two relational structures.
To check whether there exists a homomorphism $f$ from $\caA$ to $\caB$,
we can use the following algorithm.
Let $k,\ell$ be integers with $k \leq \ell$.
For each subset $U \subseteq A, |U| \leq \ell$,
we keep track of a set $\caS_{U}$ of tuples corresponding to maps from $A|_U$ to $B$.
First, we initialize $\caS_U$ to the set of solutions to the partial instance $\caA|_{U}$.
Then, for each subsets $U,U' \subseteq A $ with $ |U \cap U'| \leq k$,
we eliminate tuples $\bfa$ in $\caS_U$ if $\bfa|_{U'}$ is not contained in $\caS_{U'}|_{U}$.
We continue this process until no update occurs.
This propagation algorithm is called \klmin~\cite{bulatov2006combinatorial}.
If $\caS_U$ becomes empty for some $U \subseteq A$,
we can conclude that $\caA$ has no homomorphism to $\caB$.
Even if no $\caS_U$ is empty when propagation stops,
$\caA$ may not have a homomorphism to $\caB$.
If $\caA$ has a homomorphism to $\caB$ in such a case,
then $\caB$ is called having \textit{width $(k,\ell)$}.

A ternary operation $f: B^3 \to B$ is called a majority if $f(x,x,y) = f(x,y,x) = f(y,x,x) = x$ for $x,y \in B$.
It is known that a relational structure $\caB$ such that the associated algebra $\Alg{\caB}$ admits a majority operation has \textit{width $(2,3)$}~\cite{jeavons1997closure}.
We say a map $f:A \to B \cup \{\bot\}$ is \textit{extendable} to a homomorphism if there exists a homomorphism $f':A\to B$ such that $f'(v) = f(v)$ whenever $f(v) \in B$.
We call a vertex $v$ a \textit{violating vertex} if $f(v) \not \in \caS_{\{v\}}$ and a pair of vertices $(v,u)$ a \textit{violating pair} if $(f(v),f(u)) \not \in \caS_{\{v,u\}}$.
It is known that the majority operation implies the following property.
\begin{lemma}[\textit{2-Helly property}, \cite{feder1998computational}]\label{lmm:strict-width}
  Let $\caB$ be a relational structure such that $\Alg{\caB}$ admits a majority operation.
  For a relational structure $\caA$ and $U \subseteq A, |U| \leq 3$,
  let $\caS_U$ be the set of tuples obtained by \ttmin running on $\caA$.
  If a map $f: A \to B \cup \{\bot\}$ is not extendable to a homomorphism,
  then there is a violating vertex $v$ or a violating pair $(v,u)$ for some $u,v \in f^{-1}(B)$.
\end{lemma}
\begin{lemma}[\cite{egri2010complexity}]\label{lmm:bi-arc-has-majority}
  Let $H$ be a bi-arc graph.
  Then, $\bbH^L$ admits a majority operation.
\end{lemma}

\begin{proof}[Proof of Lemma~\ref{lmm:sublinear-upper}]
  Let $(G, L, \biw, f)$ be an input for $\LHOM{H}$.
  From Lemmas~\ref{lmm:strict-width} and~\ref{lmm:bi-arc-has-majority},
  we can assume the 2-Helly property.
  Our algorithm is described below.
  \begin{algorithm}
    \caption{$\LHOM{H}$ tester for a bi-arc graph $H$}
    \label{alg:majority}
    \begin{algorithmic}[1]
      \STATE Run \ttmin and let $\caS_U$ be the set of tuples obtained for $U \subseteq V(G), |U| \leq 3$.
      \STATE Let $X$ be a set of $\Theta(1/\epsilon)$ vertices chosen according to $\biw$
      \IF{$\exists v \in X$ such that $f(v) \not \in \caS_{\{v\}}$}
      \STATE Reject the input. \label{line:scalar}
      \ENDIF
      \STATE Let $Y_1,Y_2$ be sets of $\Theta(\sqrt{n/\epsilon})$ vertices chosen according to $\biw$.
      \IF{$\exists v \in Y_1, u \in Y_2$ such that $(f(v),f(u)) \not \in \caS_{\{v,u\}}$}
      \STATE Reject the input. \label{line:pair}
      \ENDIF
      \STATE Accept the input.
    \end{algorithmic}
  \end{algorithm}

  Note that \ttmin updates $\caS_U$ using not only the graph $G$ but also the list constraint $L$.
  The query complexity is clearly $O(\sqrt{n/\epsilon})$.
  It is clear that the tester always accepts when $f$ is a list-homomorphism.
  
  Suppose that $\dist_H(f) \geq \epsilon$.
  Let $U$ be a subset of $V(G)$ with maximum weight such that the partial homomorphism $f|_{U}$ is extendable to a list-homomorphism.
  Clearly, we have $\biw(U) + \epsilon \leq 1$.
  Let $v$ be a vertex in $V(G) \setminus U$.
  From the maximality of $U$,
  we cannot extend $f|_{U \cup \{v\}}$ to a list-homomorphism.
  Thus, from the 2-Helly property,
  $v$ is a violating vertex or $(v,u)$ is a violating pair for some $u \in U$.
  Let $A$ be the set of violating vertices in $V(G) \setminus U$,
  and $B = V(G) \setminus (U \cup A)$.
  Note that for any $v \in B$,
  a pair $(v,u)$ is a violating pair for some $u \in U$.
  Since $A \cup B = V(G) \setminus U$,
  we have $\biw(A) + \biw(B) \geq 1  - \biw(U) \geq \epsilon$.

  When $\biw(A) \geq \epsilon / 2$,
  we reject at Line~\ref{line:scalar} with probability at least $1-(1-\epsilon/2)^{\Theta(1/\epsilon)} \geq 2/3$.
  
  Suppose that $\biw(B) \geq \epsilon / 2$.
  For a subset $B' \subseteq B$,
  we define $N(B') = \{u \in U \mid \exists v \in B', (v,u)\text{ is a violating pair}\}$.
  Note that $f|_{U'}$ is extendable to a list-homomorphism where $U' = (U \setminus N(B')) \cup B'$.
  Thus, we must have $\biw(B') \leq \biw(N(B'))$ from the maximality of $U$.

  Let $q = |Y_1| = |Y_2|$ and $c > 0$ be a parameter.
  Let $B' = Y_1 \cap B$ and $F_1$ be the event that $\biw(B') \leq \frac{\epsilon q}{c n}$.
  By choosing $c$ large enough,
  we have $\Pr[F_1] \leq \frac{1}{10}$ from Chernoff's bound.
  Note that $\biw(N(B')) \geq \biw(B') \geq \frac{\epsilon q}{c n}$.
  Let $F_2$ be the event that no violated edge is detected.
  Then, $\Pr[F_2] \leq \Pr[F_1] + \Pr[F_2 \mid \overline{F_1}] \leq \frac{1}{10} + (1 - \frac{\epsilon q}{c n})^{q} \leq \frac{1}{3}$ by choosing the hidden constant in $q = \Theta(\sqrt{n/\epsilon})$ large enough.
\end{proof}
Note that we only use the fact that $\bbH^L$ admits a majority operation.
Thus, our algorithm can be also used to testing homomorphisms to other \textsf{CSPs} admitting majority operations, e.g., \twSAT.

\section{Graphs Not Testable with a Sublinear Number of Queries}\label{sec:sublinear-lower}
In this section, we prove the following, which is ``only if'' part of Theorem~\ref{thr:sublinear-list}.
\begin{lemma}\label{lmm:sublinear-lower}
  If a graph $H$ is not a bi-arc graph,
  then testing $\LHOM{H}$ requires $\Omega(n)$ queries.
\end{lemma}

To prove Lemma~\ref{lmm:sublinear-lower},
we make use of a sequence of reductions.
First, we define reductions used in this section.
\begin{definition}
  Let $\caA$ and $\caB$ be relational structures.
  We say that there is a (randomized) gap-preserving local reduction from $\caB$ to $\caA$ if there exist functions $t_1(n,m),t_2(n,m)$ and constants $c_1,c_2$ satisfying the following:
  there exists a (randomized) construction such that 
  given an input $(\caJ,\biw_J,f_J)$ for $\HOM{\caB}$,
  it generates an input $(\caI,\biw_I,f_I)$ for $\HOM{\caA}$ such that
  \begin{enumerate}
    \setlength{\itemsep}{0pt}
  \item \label{item:size} $| \caI | \leq t_1 (| \caJ |, \|\caJ\|)$,
  \item \label{item:edge-size} $\|\caI\| \leq t_2(|\caJ|, \|\caJ\|)$,
  \item \label{item:hom} if $f_J$ is a homomorphism, then $f_I$ is also a homomorphism,
  \item \label{item:far} if $\dist_{\caB}(f_J) \geq \epsilon$,
    then $\Pr[\dist_{\caA}(f_I) \geq c_1\epsilon] \geq 9/10$, and
  \item \label{item:comp} we can compute $f_I(v)$ for any $v \in I$ by querying $f_J$ at most $c_2$ times.
  \end{enumerate}
\end{definition}

\begin{lemma}\label{lmm:reduction}
  Let $\caA$ be a relational structure such that there exists a tester for $\HOM{\caA}$ with query complexity $q(n,m,\epsilon)$.
  If there exists a gap-preserving local reduction from a relational structure $\caB$ to $\caA$,
  there exists a tester for $\HOM{\caB}$ with query complexity $O(q(t_1(n,m),t_2(n,m),O(\epsilon)))$.
\end{lemma}
\begin{proof}
  Let $(\caJ,\biw_J,f_J)$ be an input for $\HOM{\caB}$.
  Let $(\caI,\biw_I,f_I)$ be the (random) input for $\HOM{\caA}$ given by the reduction.

  We run $\caA$ with an error parameter $c_1\epsilon$.
  If $f_J$ is a homomorphism,
  then $\caA$ accepts with probability at least $2/3$.
  If $f_J$ is $\epsilon$-far from homomorphisms,
  then $\caA$ rejects with probability at least $9/10 \cdot 2/3 = 3/5$.
  In both cases,
  we can increase the probability by running $\caA$ a constant number of times and take the majority of outputs.

  Since we can compute the value of $f_I(v)$ by querying $f_J$ at most $c_2$ times,
  the number of queries to $f_J$ is at most $O(c_2 q(t_1(n,m),t_2(n,m),c_1\epsilon))$.
\end{proof}

For an algebra $\bbA$,
we say that $\HOM{\bbA}$ is testable with $q(n,m,\epsilon)$ queries if,
for any relational structure $\caA \in \Str{\bbA}$, 
$\HOM{\caA}$ is testable with $q(n,m,\epsilon)$ queries.

\begin{lemma}\label{lmm:closed-under-algebra}
  Let $\caA$ be a relational structure such that $\HOM{\caA}$ is testable with $q(n,m,\epsilon)$ queries.
  Then, $\HOM{\bbA}$ is also testable with $O(1/\epsilon + q(O(n+m),O(m),O(\epsilon)))$ queries.
\end{lemma}
\begin{proof}
  Let $\caB$ be a relational structure in $\Str{\bbA}$.
  Then, each relation of $\caB$ is obtained from relations of $\caA$ in finitely many steps by using the following constructions~\cite{bodnarchuk1969galois1,bodnarchuk1969galois2}:
  \begin{enumerate}
    \setlength{\itemsep}{0pt}
  \item removing a relation,
  \item adding a relation obtained by permuting the variables of a relation,
  \item adding the intersection of two relations of the same arity,
  \item adding the product of two relations,
  \item adding the equality relation, and
  \item adding a relation obtained by projecting an $n$-ary relation to its first $n-1$ variables.
  \end{enumerate}
  It thus suffices to prove that $\caB$ is testable if $\caB$ is obtained by any of those constructions from $\caA$.
  To this end, we will give gap-preserving local reductions from $\caB$ to $\caA$ with $t_1(n,m) \leq n+m, t_2(n,m) \leq 2m, c_1 = c_2 = O(1)$.
  (For Case~5, we need reprocessing that costs $O(1/\epsilon)$ queries.)
  Let $(\caJ,w_J,f_J)$ be an input of $\HOM{\caB}$.
  Then, we construct another input $(\caI,w_I,f_I)$ of $\HOM{\caA}$ so that the construction satisfies conditions of gap-preserving local reductions.
  Since checking conditions~\eqref{item:size},\eqref{item:edge-size},~\eqref{item:hom} and~\eqref{item:comp} are straightforward,
  we will only check the condition~\eqref{item:far}.
  For any case below,
  we define $\overline{f}_I:I \to A$ as the homomorphism closest to $f_I$.
  Then, we will construct a homomorphism $\widetilde{f}_J:J \to B$ using $\overline{f}_I$ and show that $\dist(f_I,\overline{f}_I)$ must be large using the fact that $\dist(f_J,\widetilde{f}_J) \geq \epsilon$.
  Then, the lemma follows by iteratively applying Lemma~\ref{lmm:reduction}.
  
  \paragraph{Case 1:}
  Let us suppose first that $\caB$ is obtained from $\caA$ by removing a relation of $\caA$.
  Let $\caI$ be the relational structure obtained from $\caJ$ by supplementing the relations of $\caJ$ by an empty relation corresponding to the relation removed from $\caA$.
  We set $w_I = w_J$ and $f_I = f_J$.
  Then, we define $\widetilde{f}_J = \overline{f}_I$.
  It is clear that $\widetilde{f}_J$ is a homomorphism from $\caJ$ to $\caB$.
  Thus, $\dist(f_I,\overline{f}_I) = \dist(f_J,\widetilde{f}_J)  \geq \epsilon$ holds.
  
  \paragraph{Case 2:}
  Let us suppose that $\caB$ is obtained from $\caA$ by adding a relation $S$ obtained from a relation $R$ of $\caA$ by permuting the variables according to a permutation $\pi$.
  Let $\caI$ be the relational structure obtained from $\caJ$ by deleting $S^{\caJ}$ and replacing $R^{\caJ}$ by $R^{\caI} = R^{\caJ} \cup S^{\caJ}_{\pi}$ where $S^{\caJ}_{\pi}$ is obtained by permuting the variables of $S^{\caJ}$ according to $\pi^{-1}$.
  We set $\biw_I = \biw_J$ and $f_I = f_J$.
  Then, we define $\widetilde{f}_J = \overline{f}_I$.
  It is clear that $\widetilde{f}_J$ is a homomorphism from $\caJ$ to $\caB$.
  Thus, $\dist(f_I,\overline{f}_I) = \dist(f_J,\widetilde{f}_J) \geq \epsilon$ holds.

  \paragraph{Case 3:}
  Let $R$ and $S$ be two relations of the same arity of $\caA$.
  Let $T$ denote the intersection of $R$ and $S$.
  Let us suppose that $\caB$ is obtained from $\caA$ by adding the relation $T$.
  Let $\caI$ be the relational structure obtained from $\caJ$ by deleting $T^\caJ$ and replacing $R^\caJ$ by $R^\caI = R^\caJ \cup T^{\caJ}$ and $S^{\caJ}$ by $S^{\caI} = S^{\caJ} \cup T^{\caJ}$.
  We set $\biw_I = \biw_J$ and $f_I = f_J$.
  Then, we define $\widetilde{f}_J = \overline{f}_I$.
  It is clear that $\widetilde{f}_J$ is a homomorphism from $\caJ$ to $\caB$.
  Thus, $\dist(f_I,\overline{f}_I) = \dist(f_J,\widetilde{f}_J) \geq \epsilon$ holds.

  \paragraph{Case 4:}
  Let $R$ and $S$ be two relations of $\caA$.
  Let $T$ denote the product of $R$ and $S$, 
  and let $\caB$ be obtained from $\caA$ by adding the relation $T$.
  Let $\caI$ be the relational structure obtained from $\caJ$ by deleting $T^\caJ$ and replacing $R^\caJ$ by $R^\caI = R^\caJ \cup T^\caJ_1$ and $S^\caJ$ by $S^\caI  = S^\caJ \cup T^\caJ_2$ where $T^\caJ_1$ (resp., $T^\caJ_2$) is the projection of $T^\caJ$ onto the variables of the $R$-part (resp., $S$-part) of $T^{\caJ}$.
  We set $\biw_I = \biw_J$ and $f_I = f_J$.
  Then, we define $\widetilde{f}_J = \overline{f}_I$.
  It is clear that $\widetilde{f}_J$ is a homomorphism from $\caJ$ to $\caB$.
  Thus, $\dist(f_I,\overline{f}_I) = \dist(f_J,\widetilde{f}_J) \geq \epsilon$ holds.

  \paragraph{Case 5:}
  Let us suppose that $\caB$ is obtained from $\caA$ by adding the equality relation.
  Let $\theta$ be the reflexive, symmetric, transitive closure of $\theta'$ where $\theta'$ is the relation of $\caJ$ corresponding to equality in $\caB$.
  Clearly, $\theta$ is an equivalent relation on $J$.
  For a variable $v \in J$,
  we define $v/_\theta$ as the corresponding $\theta$-block.
  For a $\theta$-block $u$,
  we define $\biw_J(u,b) = \sum\limits_{v \in u: f_J(v) = b}\biw_J(v)$,
  $\biw_J(u) = \sum\limits_{b \in B}\biw_J(u,b)$,
  $\biw_J^{\maj}(u) = \max\limits_{b \in B} \biw_J(u,b)$,
  and $f_J^{\maj}(u) = \mathop{\mathrm{argmax}}\limits_{b \in B} \biw_J(u,b)$.
  
  Before we reduce the problem to $\HOM{\caA}$,
  we run the following algorithm first:
  Pick a set of $\Theta(1/\epsilon)$ variables according to $\biw$ and check whether those variables obey $\theta$.
  Let $\widehat{f}_J: J \to B$ be the map such that $\widehat{f}_J(v) = f_J^{\maj}(v/_\theta)$.
  Note that $\widehat{f}_J$ is the map closest to $f_J$ obeying $\theta$.
  It is easy to see that the above algorithm always accepts when $f_J$ is a homomorphism from $\caJ$ to $\caB$ and the query complexity is $O(1/\epsilon)$.
  
  Suppose that $\dist(f_J,\widehat{f}_J) \geq \epsilon / 20$.
  This indicates $\sum_{u:\theta\text{-block}}(\biw_J(u)-\biw_J^{\maj}(u)) \geq \epsilon / 20$.
  It is not hard to show that,
  in such a case,
  the algorithm above rejects $\widehat{f}_J$ with probability at least $2/3$.
  Thus, 
  we assume that $\dist(f_J,\widehat{f}_J) = \sum_{u:\theta\text{-block}}(\biw_J(u)-\biw_J^{\maj}(u)) < \epsilon / 20$ in what follows.

  Now, we define the input $(\caI,\biw_I,f_I)$ for $\HOM{\caA}$.
  We define a structure $\caI$ so that its base set is the set of $\theta$-blocks,
  and its relations are defined as follows:
  For each relation $R^{\caJ}$ of $\caJ$ distinct from $\theta'$,
  we define a relation $R^{\caI}$ on the $\theta$-blocks by stipulating that 
  \begin{eqnarray*}
    (u_1,\ldots,u_r) \in R^{\caI} \text{ iff } \exists v_1 \in u_1,\ldots,\exists v_r \in u_r \text{ such that } (v_1,\ldots,v_r) \in R^{\caJ}.
  \end{eqnarray*}
  Also, we set $\biw_I(u) = \biw_J(u)$.
  We define $f_I(u) = f_J(v)$ where $v$ is a variable in $u$ randomly chosen according to $\biw_J|_{u}$,
  i.e., $v$ is chosen with probability $\biw_J(v)/\biw_J(u)$.

  We define $\dist(f_J,f_I) = \sum\limits_{u:\theta\text{-block}}\sum\limits_{v \in u: f_I(u) \neq f_J(v)}\biw_J(v)$.
  For a $\theta$-block $u$, 
  we define $\dist_u(f_J,f_I) = \sum\limits_{v \in u: f_I(u) \neq f_J(v)}\frac{\biw_J(v)}{\biw_J(u)}$.
  It is clear that $\dist(f_J,f_I) = \sum\limits_{u:\theta\text{-block}}\biw_J(u)\dist_u(f_J,f_I)$.
  
  We have 
  \begin{eqnarray*}
    \E[\dist_u(f_J,f_I)] 
    &=&
    \sum_{b \in B}\frac{\biw_J(u,b)}{\biw_J(u)}\left(1 - \frac{\biw_J(u,b)}{\biw_J(u)}\right)\\
    &\leq&
    \frac{\biw_J^{\maj}(u)}{\biw_J(u)}\left(1 - \frac{\biw_J^{\maj}(u)}{\biw_J(u)}\right) + \left(1 - \frac{\biw_J^{\maj}(u)}{\biw_J(u)}\right)\cdot 1 \\
    &\leq&
    2\left(1- \frac{\biw_J^{\maj}(u)}{\biw_J(u)}\right).
  \end{eqnarray*}
  Thus,
  \begin{eqnarray*}
    \E[\dist(f_J,f_I)] = \E[\sum_{u:\theta\text{-block}}\biw_J(u)\dist_u(f_J,f_I)] 
    \leq 
    \sum_{u:\theta\text{-block}} 2(\biw_J(u)- \biw_J^{\maj}(u))
    < 
    \frac{\epsilon}{10}.
  \end{eqnarray*}

  Also, we have 
  \begin{eqnarray*}
    &&
    \Var[\dist_u(f_J,f_I)] \\
    &=&
    \E[(\dist_u(f_J,f_I) - \E[\dist_u(f_J,f_I)])^2]\\
    &\leq&
    \frac{\biw_J^{\maj}(u)}{\biw_J(u)}\left(1-\frac{\biw_J^{\maj}(u)}{\biw_J(u)} - \E[\dist_u(f_J,f_I)]\right)^2  + \left(1 - \frac{\biw_J^{\maj}(u)}{\biw_J(u)}\right)\cdot 1 \\
    &=&
    \left(1-\frac{\biw_J^{\maj}(u)}{\biw_J(u)}\right)^2  + \left(1 - \frac{\biw_J^{\maj}(u)}{\biw_J(u)}\right) 
    \leq
    2\left(1- \frac{\biw_J^{\maj}(u)}{\biw_J(u)}\right).
  \end{eqnarray*}
  Thus,
  \begin{eqnarray*}
    \Var[\dist(f_J,f_I)] = \Var[\sum_{u:\theta\text{-block}}\biw_J(u)\dist_u(f_J,f_I)]
    \leq 
    \sum_{u:\theta\text{-block}} 2(\biw_J(u)- \biw_J^{\maj}(u))
    \leq 
    \frac{\epsilon}{10}.
  \end{eqnarray*}
  Thus, from Chebyshev's inequality,
  $\Pr[\dist(f_J,f_I) \geq \epsilon/2] \leq 1/16$.

  We check the condition~\eqref{item:far}.
  We define $\widetilde{f}_J$ as $\widetilde{f}_J(v) =  \overline{f}_I(v/_\theta)$.
  It is clear that $\widetilde{f}_J$ is a homomorphism from $\caJ$ to $\caB$.
  Since we have $\dist(f_J,f_I) + \dist(f_I,\overline{f}_I) \geq \dist(f_J,\widetilde{f}_J) \geq \epsilon$,
  it follows that $\Pr[\dist(f_I,\overline{f}_I) \geq \epsilon / 2] \geq 15/16$.

  \paragraph{Case 6:}
  Let us suppose that $\caB$ is obtained from $\caA$ by adding the projection $S$ of an $r$-ary relation $R$ of $\caA$ to its first $r-1$ variables.
  Let  $\caI$ be the relational structure with the base set $J$ extended by a new element for each $(r-1)$-tuple in $S^\caJ$.
  The relations of $\caI$ are those of $\caJ$,
  except that $S^\caJ$ is removed and $R^\caJ$ is replaced by $R^\caI=R^\caJ \cup S^\caJ_{r-1}$ 
  where $S^\caJ_{r-1}$ is obtained from $S^\caJ$ by extending every $(r-1)$-tuple of $S^\caI$ with the corresponding new element in the base set of $I$.
  Note that $|\caI| \leq |\caJ| + \|\caJ\|$ and $\|\caI\| = \|\caJ\|$.
  We set $\biw_I(v) = \biw_J(v)$ for a variable $v \in J$ is and $\biw_I(v) = 0$ for a variable $v$ corresponding to $(r-1)$-tuple in $S^\caJ$.
  We define $f_I$ as follows:
  We set $f_I(v) = f_J(v)$ for $v \in J$ and arbitrary value for a variable $v$ corresponding to an $(r-1)$-tuple  in $ S^\caJ$. (Indeed, we do not have to care about those values since there weights are zero.)
  Then, we define $\widetilde{f}_J(v) = \overline{f}_I(v)$ for $v \in J$.
  It is easy to check that $\widetilde{f}_J$ is a homomorphism,
  and $\dist(f_I,\overline{f}_I) \geq \dist(f,\widetilde{f}_J) \geq \epsilon$.
\end{proof}

We introduce some notions related to algebras.
Let $\bbA = (A; F)$ be an algebra.
A set $B \subseteq A$ is a \textit{subuniverse} of $\bbA$ if for every basic
operation $f \in F$ restricted to $B$ has all the results in $B$.
For a nonempty subuniverse $B$ of an algebra $\bbA$, $f|_{B}$ is the
restriction of $f$ to $B$.
The algebra $\bbB = (B,F|_{B})$ where $F|_{B} = \{f|_{B} \mid f \in F\}$
is a \textit{subalgebra} of $\bbA$.
Algebras $\bbA,\bbB$ are of the \textit{same type} if they have the same number of basic operations and corresponding operations have equal arities.
Given algebras $\bbA,\bbB$ of the same type, a \textit{product}
$\bbA\times \bbB$ is the algebra with the same type as $\bbA$ and $\bbB$
with universe $A \times B$ and basic operations computed
coordinate-wise.
An equivalence relation $\theta$ on $A$ is called a \textit{congruence} of an algebra $\bbA$ if $\theta$ is a subalgebra of $\bbA \times \bbA$.
Given a congruence $\theta$ on $A$, we can form the \textit{homomorphic image} $\bbA /_{\theta}$, whose elements are the equivalence classes
of $\bbA$ and the basic operations are defined so that the natural
projection mapping is a homomorphism $\bbA \to \bbA/_{\theta}$.
If $\bbA$ is idempotent, $\theta$-classes are subuniverses of $\bbA$.
It is known that a relational structure $\caA$ is in $\Str{\bbA}$ if each relation in $\caA$ is a subuniverse of a finite power of $\bbA$.

A \textit{variety} is a class of algebras of the same type closed under formation of subalgebras,
homomorphic images and finite products.
For any algebra $\bbA$,
there is a smallest variety containing $\bbA$, denoted by $\caV(\bbA)$ and called the \textit{variety generated} by $\bbA$.
It is well known that any variety is generated by an algebra and that any member of $\caV(\bbA)$ is a homomorphic image of a subalgebra of a power of $\bbA$.

\begin{lemma}\label{lmm:closed-under-variety}
  Let $\bbA$ be an algebra such that $\HOM{\bbA}$ is testable with $q(n,m,\epsilon)$ queries.
  Then, for any finite algebra $\bbB \in \caV(\bbA)$,
  $\HOM{\bbB}$ is also testable with $O(q(O(n),O(m),O(\epsilon)))$ queries.
\end{lemma}
\begin{proof}
  It suffices to show that every subalgebra, homomorphic image and finite power of $\bbA$ is testable with $O(q(O(n),O(m),O(\epsilon)))$ queries.
  Let $\bbB$ be a subalgebra, a homomorphic image, or a finite power of $\bbA$ and let $\caB$ be a relational structure on $B$ such that the relations of $\caB$ are subalgebras of finite powers of $\bbB$.
  From Lemma~\ref{lmm:reduction},
  it suffices to show a gap-preserving local reduction from $\caB$ to $\caA$.
  We follow the approach similar to the proof of Lemma~\ref{lmm:closed-under-algebra}.
  Given an input $(\caJ,\biw_J,f_J)$ for $\HOM{\caB}$,
  we define another structure $(\caI,\biw_I,f_I)$ for $\HOM{\caA}$.
  Then, we show that the construction satisfies the conditions of gap-preserving local reductions.
  Since checking conditions~\eqref{item:size},~\eqref{item:edge-size},~\eqref{item:hom}, and~\eqref{item:comp} are straightforward,
  we will only check the condition~\eqref{item:far}.
  For any case below, we define $\overline{f}_I:I\to A$ as a homomorphism closest to $f_I$.
  Then, we will construct a homomorphism $\widetilde{f}_J:J \to B$ from $\overline{f}_I$ and show that $\dist(f_I,\overline{f}_I)$ must be large by using the fact that $\dist(f_J,\widetilde{f}_J) \geq \epsilon$.
  
  Suppose first that $\bbB$ is a subalgebra of $\bbA$.
  Let $\caA$ be the relational structure whose base set is $A$ and whose relations are all the relations of $\caB$ and $B$ as a unary relation.
  Notice that the relations of $\caA$ are subalgebras of finite powers of $\bbA$,
  and $\caA \in \Str{\bbA}$.
  We take $\caI$ to be $\caJ$ with all of its relations adding $B$ as a unary relation.
  In particular, $I = J$.
  Then, we define $\biw_I = \biw_J$ and $f_I = f_J$.
  Due to the unary relations,
  $\overline{f}_I(v) \in B$ must hold for every $v \in I$.
  Thus, $\overline{f}_I$ is also a homomorphism from $\caJ$ to $\caB$.
  Thus, $\dist_\caA(f) = \dist(f_I,\overline{f}_I) \geq \epsilon$.

  Secondly, suppose that $\bbB$ is a homomorphic image of $\bbA$ under the homomorphism $h:A \to B$.
  This time, let $\caA$ be the relational structure whose base set is $A$ and whose relations are the preimages under the homomorphism $h$ of all the relations of $\caB$.
  Notice that the relations of $\caA$ are subalgebras of finite powers of $\bbA$,
  and $\caA \in \Str{\bbA}$.
  We take $\caI$ to be the relational structure whose base set is $J$ and whose relations are the preimages under the homomorphism $h$ of all the relations of $\caI$.
  We define $\biw_I = \biw_J$ and $f_I:I \to A$ so that $f_I(v)$ is any element in $h^{-1}(f_J(v))$.
  Note that $h\circ \overline{f}_I$ is a homomorphism from $I$ to $B$.
  From the construction, 
  we have $\dist(f_I,\overline{f}_I) \geq \dist(h\circ f_I,h\circ \overline{f}_I) = \dist(f_J,h\circ \overline{f}_I) \geq \epsilon$.

  Finally, suppose that $\bbB=\bbA^k$.
  Let $\caA$ be the relational structure with the following relations:
  If $R$ is an $s$-ary relation of $\caB$,
  define $R_0$ to be the $sk$-ary relation such that,
  if $(b_1,\ldots,b_s) \in R$ with $b_i=(a_{1,i},\ldots,a_{k,i})$,
  we put the $sk$-tuple $(a_{1,1},\ldots,a_{1,s},\ldots,a_{k,1},\ldots,a_{k,s})$ in $R_0$.
  Note that the $sk$-ary relations obtained in this way are subalgebras of finite powers of $\bbA$,
  and $\caA \in \Str{\bbA}$.
  We take $\caI$ to be the union of $k$ disjoint copies of $J$ with one $sk$-ary relation for each $s$-ary relation of $\caJ$.
  An $sk$-tuple in the new relation on $\caI$ is formed by the $k$ copies of an $s$-tuple in the old relation on $\caJ$,
  that is,
  if $(x_1,\ldots,x_s)$ is in the old relation on $\caJ$ and $x_{i,j}$ is the $i$-th copy of $x_j$ in $J$,
  then $(x_{1,1},\ldots,x_{1,s},\ldots,x_{k,1},\dots,x_{k,s})$ is in the new relation.
  We define $\biw_I(x_i) = \biw_J(x) / k$ if $x_i$ is a copy of $x$.
  For a map $f_J:J \to B$,
  we define $f_I:I \to A$ as follows:
  If $x_{i,j}$ is the $i$-th copy of $x_j$,
  we define $f_I(x_{i,j})$ as the $i$-th element of $f_J(x_j)$.
  We make $\widetilde{f}_J: J \to B$ from $\overline{f}_I$ by $\widetilde{f}_J(x_j) = (\overline{f}_I(x_{1,j}),\ldots,\overline{f}_I(x_{n,j}))$.
  Clearly, $\widetilde{f}_J$ is a homomorphism.
  Thus, $\dist(f_I,\overline{f}_I) \geq \dist(f_J,\widetilde{f}_J) / k \geq \epsilon / k$.
\end{proof}

We define $\mathsf{3LIN} = (\{0,1\}; \{(0,0,0),(0,1,1),(1,0,1),(1,1,0)\},\{(0,0,1),(0,1,0),(1,0,0),(1,1,1)\})$ as the relational structure expressing a system of linear equations over $\bbF_2$ such that each equation has arity three.
\begin{lemma}[\cite{ben2006some}]\label{lmm:lin-hard}
  Testing $\HOM{\mathsf{3LIN}}$ requires $\Omega(n)$ queries even if $m = O(n)$.
\end{lemma}

We introduce the notion of \textit{type set} of an algebra and of a variety.
Roughly speaking, the type set of a finite algebra is a subset of the set ${1, 2, 3, 4, 5}$ whose elements are called types and correspond to certain classes of algebras: 1 to unary algebras, 2 to vector spaces over finite fields, 3 to Boolean algebras, 4 to distributive lattices and 5 to semilattices.
See~\cite{hobby1988structure} for details.
If $\caV$ is a variety, the type set of $\caV$ is the union of the type sets of the finite algebras in $\caV$.
We say that an algebra or variety admits (omits) type $i$ when $i$ is (is not) in its type set. 

\begin{lemma}[\cite{krokhin2005complexity}]\label{lmm:hard-variety}
  Let $\bbA$ be an idempotent algebra such that $\caV(\bbA)$ admits type $1$.
  There exists an algebra $\bbB\in \caV(\bbA)$ such that $\mathsf{3LIN} \in \Str{\bbB}$.
\end{lemma}

\begin{lemma}[\cite{feder2003bi}]\label{lmm:bi-arc-hard}
  Let $H$ be a non bi-arc graph.
  Then, $\caV(\bbH^L)$ admits type $1$.
\end{lemma}

\begin{proof}[Proof of Lemma~\ref{lmm:sublinear-lower}]
  Assume that $\LHOM{H}$ is testable with $o(n)$ queries when $m = O(n)$.
  Then, from Lemma~\ref{lmm:closed-under-algebra},
  $\HOM{\bbH^L}$ is testable with $o(n)$ queries.
  Then, from Lemma~\ref{lmm:closed-under-variety},
  $\HOM{\bbH'}$ is testable with $o(n)$ queries for any $\bbH' \in \caV(\bbH^L)$.
  However, 
  $\bbH^L$ is idempotent and $\caV(\bbH^L)$ admits type $1$ from Lemma~\ref{lmm:bi-arc-hard}.
  Thus, $\mathsf{3LIN} \in \caV(\bbH^L)$ from Lemmas~\ref{lmm:hard-variety},
  and testing $\mathsf{3LIN}$ requires $\Omega(n)$ queries from Lemma~\ref{lmm:lin-hard}.
  Contradiction.
\end{proof}


\end{document}